\documentclass[11pt,letterpaper]{article}

\usepackage{graphicx,subfigure}

\def\showauthornotes{1}

\def\showkeys{0}
\def\showdraftbox{1}
\def\showcolorlinks{1}
\def\usemicrotype{1}
\def\showfixme{1}

\def\stocmode{0}
\def\jamesmode{0}
\def\arxivmode{0}
\def\fastmode{0}

\newcommand{\cost}{\mathsf{cost}}

\newcommand{\llangle}{\left\langle}
\newcommand{\rrangle}{\right\rangle}
\newcommand{\vvW}{\vvmathbb{W}}

\ifnum\fastmode=0
\usepackage{etex}
\fi

\ifnum\fastmode=0
\usepackage[l2tabu, orthodox]{nag}
\fi

\usepackage{xspace,enumerate}

\usepackage[dvipsnames]{xcolor}

\ifnum\fastmode=0
\usepackage[T1]{fontenc}
\usepackage[full]{textcomp}
\fi

\usepackage[american]{babel}
\usepackage{mathtools}

\usepackage{amsthm}

\newtheorem{theorem}{Theorem}[section]
\newtheorem*{theorem*}{Theorem}

\newtheorem*{proposition*}{Proposition}
\newtheorem{lemma}[theorem]{Lemma}
\newtheorem*{lemma*}{Lemma}

\newtheorem*{conjecture*}{Conjecture}

\newtheorem*{fact*}{Fact}

\newtheorem*{exercise*}{Exercise}

\newtheorem*{hypothesis*}{Hypothesis}

\theoremstyle{definition}

\newtheorem{exercise-easy}[theorem]{Exercise}
\newtheorem{exercise-med}[theorem]{Exercise}
\newtheorem{exercise-hard}[theorem]{Exercise$^\star$}

\newtheorem*{claim*}{Claim}

\newtheorem*{remark*}{Remark}

\newtheorem*{observation*}{Observation}

\usepackage[
letterpaper,
top=0.7in,
bottom=0.9in,
left=1in,
right=1in]{geometry}

\ifnum\arxivmode=0
\usepackage{newpxtext} % T1, lining figures in math, osf in text
\usepackage{textcomp} % required for special glyphs
\usepackage[varg,bigdelims]{newpxmath}
\usepackage[scr=rsfso]{mathalfa}% \mathscr is fancier than \mathcal
\usepackage{bm} % load after all math to give access to bold math
\let\mathbb\varmathbb
\fi

\ifnum\arxivmode=1
\usepackage[varg]{pxfonts} % varg - uses nicer g,v,y,w
\usepackage{textcomp} % required for special glyphs
\usepackage[scr=rsfso]{mathalfa}% \mathscr is fancier than \mathcal
\usepackage{bm} % load after all math to give access to bold math
\fi

\ifnum\showkeys=1
\usepackage[color]{showkeys}
\fi

\makeatletter
\@ifpackageloaded{hyperref}
{}
{\usepackage[pagebackref]{hyperref}}

\definecolor{bleudefrance}{rgb}{0.01, 0.1, 1.0}
\definecolor{azure}{rgb}{0.0, 0.5, 1.0}

\ifnum\showcolorlinks=1
\hypersetup{
pagebackref=true,
colorlinks=true,
urlcolor=blue,
linkcolor=blue,
citecolor=OliveGreen}
\fi

\ifnum\showcolorlinks=0
\hypersetup{
pagebackref=true,
colorlinks=false,
pdfborder={0 0 0}}
\fi

\usepackage{prettyref}

\newcommand{\savehyperref}[2]{\texorpdfstring{\hyperref[#1]{#2}}{#2}}

\newrefformat{eq}{\savehyperref{#1}{\textup{(\ref*{#1})}}}
\newrefformat{lem}{\savehyperref{#1}{Lemma~\ref*{#1}}}
\newrefformat{def}{\savehyperref{#1}{Definition~\ref*{#1}}}
\newrefformat{thm}{\savehyperref{#1}{Theorem~\ref*{#1}}}
\newrefformat{cor}{\savehyperref{#1}{Corollary~\ref*{#1}}}
\newrefformat{cha}{\savehyperref{#1}{Chapter~\ref*{#1}}}
\newrefformat{sec}{\savehyperref{#1}{Section~\ref*{#1}}}
\newrefformat{app}{\savehyperref{#1}{Appendix~\ref*{#1}}}
\newrefformat{tab}{\savehyperref{#1}{Table~\ref*{#1}}}
\newrefformat{fig}{\savehyperref{#1}{Figure~\ref*{#1}}}
\newrefformat{hyp}{\savehyperref{#1}{Hypothesis~\ref*{#1}}}
\newrefformat{alg}{\savehyperref{#1}{Algorithm~\ref*{#1}}}
\newrefformat{rem}{\savehyperref{#1}{Remark~\ref*{#1}}}
\newrefformat{item}{\savehyperref{#1}{Item~\ref*{#1}}}
\newrefformat{step}{\savehyperref{#1}{step~\ref*{#1}}}
\newrefformat{conj}{\savehyperref{#1}{Conjecture~\ref*{#1}}}
\newrefformat{fact}{\savehyperref{#1}{Fact~\ref*{#1}}}
\newrefformat{prop}{\savehyperref{#1}{Proposition~\ref*{#1}}}
\newrefformat{prob}{\savehyperref{#1}{Problem~\ref*{#1}}}
\newrefformat{claim}{\savehyperref{#1}{Claim~\ref*{#1}}}
\newrefformat{relax}{\savehyperref{#1}{Relaxation~\ref*{#1}}}
\newrefformat{red}{\savehyperref{#1}{Reduction~\ref*{#1}}}
\newrefformat{part}{\savehyperref{#1}{Part~\ref*{#1}}}
\newrefformat{ex}{\savehyperref{#1}{Exercise~\ref*{#1}}}
\newrefformat{property}{\savehyperref{#1}{Property~\ref*{#1}}}

\newcommand{\Sref}[1]{\hyperref[#1]{\S\ref*{#1}}}

\usepackage{nicefrac}
\ifnum\fastmode=0
\ifnum\usemicrotype=1
\usepackage{microtype}
\fi
\fi

\ifnum\showauthornotes=2
\newcommand{\mynotes}[1]{{\sffamily\small\color{teal}{#1}}\medskip}
\newcommand{\Authornote}[2]{{\sffamily\small\color{blue}{[#1: #2]}}\medskip}
\newcommand{\Authornotecolored}[3]{{\sffamily\small\color{#1}{[#2: #3]}}}
\newcommand{\Authorcomment}[2]{{\sffamily\small\color{gray}{[#1: #2]}}}
\newcommand{\Authorstartcomment}[1]{\sffamily\small\color{gray}[#1: }

\newcommand{\Authorfnote}[2]{\footnote{\color{red}{#1: #2}}}
\newcommand{\Authorfixme}[1]{\Authornote{#1}{\textbf{??}}}
\newcommand{\Authormarginmark}[1]{\marginpar{\textcolor{red}{\fbox{\Large #1:!}}}}
\newcommand{\myexplain}[1]{{\sffamily\small\color{red}{\noindent [Explanation:\medskip\newline \begin{quote}#1\hfill]\end{quote}}}\medskip}
\newcommand{\explain}[1]{{\sffamily\small\color{red}{#1}}\medskip}

\else

\newcommand{\mynotes}[1]{}
\newcommand{\Authornote}[2]{}
\newcommand{\Authornotecolored}[3]{}
\newcommand{\Authorcomment}[2]{}
\newcommand{\Authorstartcomment}[1]{}

\newcommand{\Authorfnote}[2]{}
\newcommand{\Authorfixme}[1]{}
\newcommand{\Authormarginmark}[1]{}
\newcommand{\myexplain}[1]{}
\newcommand{\explain}[1]{}

\ifnum\showauthornotes=1
\renewcommand{\myexplain}[1]{{\sffamily\small\color{red}{\noindent \begin{quote}{\bf Explanation:} \medskip\newline #1\end{quote}}}\medskip}
\fi

\fi

\ifnum\showfixme=0

\fi

\usepackage{boxedminipage}

\newcommand{\Esymb}{\mathbb{E}}

\DeclareMathOperator*{\E}{\Esymb}

\newcommand{\textparen}[1]{\text{(#1)}}

\ifx\because\undefined
\newcommand{\because}[1]{\textparen{because #1}}
\else
\renewcommand{\because}[1]{\textparen{because #1}}
\fi

\newcommand{\seteq}{\mathrel{\mathop:}=}

\newcommand\bdot\bullet

\ifx\mathds\undefined % use double stroke fonts if available
\newcommand{\Ind}{\mathbb I}
\else
\newcommand{\Ind}{\mathds 1}
\fi

\DeclareMathOperator{\dist}{dist}

\newcommand{\Z}{\mathbb Z}

\newcommand{\R}{\mathbb R}

\newcommand{\cL}{\mathcal L}

\newcommand{\cT}{\mathcal T}

\renewcommand{\leq}{\leqslant}

\renewcommand{\geq}{\geqslant}

\ifnum\showdraftbox=1

\else

\fi

\let\epsilon=\varepsilon

\numberwithin{equation}{section}

\newcommand\MYcurrentlabel{xxx}

\newcommand{\MYstore}[2]{%
  \global\expandafter \def \csname MYMEMORY #1 \endcsname{#2}%
}

\newcommand{\MYload}[1]{%
  \csname MYMEMORY #1 \endcsname%
}

\newcommand{\MYnewlabel}[1]{%
  \renewcommand\MYcurrentlabel{#1}%
  \MYoldlabel{#1}%
}

\newcommand{\MYdummylabel}[1]{}

\newcommand{\torestate}[1]{%
  \let\MYoldlabel\label%
  \let\label\MYnewlabel%
  #1%
  \MYstore{\MYcurrentlabel}{#1}%
  \let\label\MYoldlabel%
}

\newcommand{\restatetheorem}[1]{%
  \let\MYoldlabel\label
  \let\label\MYdummylabel
  \begin{theorem*}[Restatement of \prettyref{#1}]
    \MYload{#1}
  \end{theorem*}
  \let\label\MYoldlabel
}

\newcommand{\restatelemma}[1]{%
  \let\MYoldlabel\label
  \let\label\MYdummylabel
  \begin{lemma*}[Restatement of \prettyref{#1}]
    \MYload{#1}
  \end{lemma*}
  \let\label\MYoldlabel
}

\newcommand{\restateprop}[1]{%
  \let\MYoldlabel\label
  \let\label\MYdummylabel
  \begin{proposition*}[Restatement of \prettyref{#1}]
    \MYload{#1}
  \end{proposition*}
  \let\label\MYoldlabel
}

\newcommand{\restatefact}[1]{%
  \let\MYoldlabel\label
  \let\label\MYdummylabel
  \begin{fact*}[Restatement of \prettyref{#1}]
    \MYload{#1}
  \end{fact*}
  \let\label\MYoldlabel
}

\newcommand{\restate}[1]{%
  \let\MYoldlabel\label
  \let\label\MYdummylabel
  \MYload{#1}
  \let\label\MYoldlabel
}

\newcommand{\addreferencesection}{
  \phantomsection
\ifnum\stocmode=0
  \addcontentsline{toc}{section}{References}
\else
  \addcontentsline{toc}{section}{References \hspace*{1in} --------- End of extended abstract ---------}
\fi

}

\let\origparagraph\paragraph
\renewcommand{\paragraph}[1]{\vspace*{-10pt}\origparagraph{#1.}}

\allowdisplaybreaks

\usepackage{paralist}

\usepackage{comment}

\let\pref=\prettyref

\newcommand{\diam}{\mathrm{diam}}

\newcommand{\vertiii}[1]{{\left\vert\kern-0.25ex\left\vert\kern-0.25ex\left\vert #1 
          \right\vert\kern-0.25ex\right\vert\kern-0.25ex\right\vert}}

\renewcommand{\Ind}{\vvmathbb{1}}
\ifnum\jamesmode=0
\fi
\let\1\Ind

\DeclareMathOperator{\argmin}{\mathrm{argmin}}
\usepackage{accents}

\usepackage{enumitem}
\usepackage[normalem]{ulem}

\renewcommand{\Z}{\vvmathbb{Z}}
\renewcommand{\R}{\vvmathbb{R}}

\begin{document}

\title{Metrical task systems on trees \\
via mirror descent and unfair gluing\footnote{A preliminary version of this paper appeared at SODA 2019.}}
\author{S\'ebastien Bubeck \\ {\small Microsoft Research} \and Michael B. Cohen \\ {\small MIT} \and James R. Lee \qquad Yin Tat Lee \\ {\small University of Washington}}

\date{}

\maketitle

\begin{abstract}
We consider metrical task systems on tree metrics, and present an $O(\mathrm{depth} \times \log n)$-competitive randomized algorithm based on the mirror descent framework introduced in our prior work on the $k$-server problem. For the special case of hierarchically separated trees (HSTs),
we
use mirror descent
to refine the standard approach based on gluing unfair metrical task systems. This yields an $O(\log n)$-competitive algorithm for HSTs, thus removing an extraneous $\log\log n$ in the bound of Fiat and Mendel (2003).
Combined with well-known HST embedding theorems,
this also gives an $O((\log n)^2)$-competitive randomized algorithm
for every $n$-point metric space.
\end{abstract}

\setcounter{tocdepth}{2}
\tableofcontents

\newpage

\section{Introduction}
Let $(X,d)$ be a finite metric space with $|X|=n > 1$. The Metrical Task Systems (MTS) problem, introduced in \cite{BLS92}, can be described as follows.
The input is a sequence $\langle c_t : X \rightarrow \R_+ : t \geq 1\rangle$ of nonnegative cost functions on the state space $X$.
At every time $t$, an online algorithm maintains
a state $\rho_t \in X$.

The corresponding cost is the
sum of a {\em service cost} $c_t(\rho_t)$ and a {\em movement cost} $d(\rho_{t-1}, \rho_t)$.
Formally, an {\em online algorithm} is a sequence of mappings
$\bm{\rho} = \langle \rho_1, \rho_2, \ldots, \rangle$
where, for every $t \geq 1$,
$\rho_t : (\R_+^X)^t \to X$ maps a sequence of cost functions $\langle c_1, \ldots, c_t\rangle$
to a state. The initial state $\rho_0 \in X$ is fixed. The {\em total cost of the algorithm $\bm{\rho}$ in servicing $\bm{c} = \langle c_t : t \geq 1\rangle$} is defined as:
\[
   \cost_{\bm{\rho}}(\bm{c}) \seteq \sum_{t \geq 1} \left[c_t\left(\rho_t(c_1,\ldots, c_t)\right) + d\left(\rho_{t-1}(c_1,\ldots, c_{t-1}), \rho_t(c_1,\ldots, c_t)\right)\right].
\]
The cost of the {\em offline optimum}, denoted $\cost^*(\bm{c})$, is the infimum of
$\sum_{t \geq 1} [c_t(\rho_t)+d(\rho_{t-1},\rho_t)]$ over {\em any} sequence
$\llangle \rho_t : t \geq 1\rrangle$ of states.
A {\em randomized online algorithm} $\bm{\rho}$  is said to be {\em $\alpha$-competitive}
if for every $\rho_0 \in X$, there is a constant $\beta > 0$ such that for all
cost sequences $\bm{c}$:
\[
\E\left[\cost_{\bm{\rho}}(\bm{c})\right] \leq \alpha \cdot \cost^*(\bm{c}) + \beta\,.
\]

\paragraph{Tree metrics} 
We will be 
primarily concerned with tree metrics:  Those metric spaces $(X,d)$ that correspond
to the shortest-path distance on a finite (rooted) tree with prescribed nonnegative edge lengths.
We refer to the {\em combinatorial depth} of such a tree
as the depth of the corresponding unweighted tree.

A {\em hierarchical separated tree} (HST) with separation $\tau > 1$ is a tree metric such that, for any edge, the diameter of the subtree rooted at the end of the edge is at most $1/\tau$ times the weight of that edge.
Such a space is referred to as a {\em $\tau$-HST metric.}
The importance of HSTs stems from the well-known fact that any metric space can be probabilistically embedded into a weighted HST of depth $O(\log n)$ and separation $\tau$ with 
distortion $O(\tau \log n)$ \cite{Bar96, FRT04, BBMN15}.
In particular, an $O(f(n))$-competitive algorithm for $\tau$-HSTs
implies an $O(f(n) \tau \log n)$ competitive algorithm for an arbitrary $n$-point metric space.

\paragraph{Contributions and related work}
For the $n$-point uniform metric, i.e., the path metric
on the leaves of a unit-weighted star,
a simple coupon-collector argument shows that the competitive ratio has to be $\Omega(\log n)$, and this is tight \cite{BLS92}. A long-standing conjecture is that this $\Theta(\log n)$ competitive ratio holds for an arbitrary $n$-point metric space.

The lower bound has almost been established \cite{BBM06,BLMN05};
for any $n$-point metric space, the competitive ratio is $\Omega(\log n / \log \log n)$.
On the other hand, a matching upper bound of $O(\log n)$ was previously
only known for weighted star metrics (this can be deduced from the analysis in \cite{BBN12}).
Our first result extends this bound to constant-depth tree metrics as follows.

\begin{theorem} \label{thm:trees}
   There is an $O(D \log n)$-competitive randomized
   algorithm for MTS on any $n$-point tree metric with combinatorial depth $D$.
\end{theorem}

The above result is obtained by an application of the mirror descent framework introduced in our prior work \cite{BCLLM18} on the $k$-server problem.
We obtain the following more precise bounds, which are referred to as {\em refined guarantees} (see, e.g., \cite[Thm. 4]{BBN10}).

For a randomized online algorithm $\bm{\rho}$ and a cost sequence $\bm{c}$, we denote respectively $S_{\bm{\rho}}(\bm{c})$ and $M_{\bm{\rho}}(\bm{c})$ for the (expected) service cost and movement cost, that is
\[
S_{\bm{\rho}}(\bm{c}) \seteq \E \sum_{t \geq 1} c_t(\rho_t) \quad \text{and}\quad M_{\bm{\rho}}(\bm{c}) \seteq \E \sum_{t\geq1}d(\rho_{t-1}, \rho_t) \,.
\]

\begin{theorem} \label{thm:refined}
Consider an $n$-point tree metric with combinatorial depth $D$. There is an online randomized algorithm $\bm{\rho}$ that achieves,
for any $\bm{c}$,
\[
S_{\bm{\rho}}(\bm{c}) \leq \cost^*(\bm{c})\,,
\]
and
\[
M_{\bm{\rho}}(\bm{c}) \leq O(D \log n) \ (\cost^*(\bm{c}) + \mathrm{diam}(X)) \,.
\]
\end{theorem}

For $n$-point HST metrics, Fiat and Mendel \cite{FM03} achieve an $O((\log n)\log \log n)$ competitive ratio,
improving on the $O((\log n)^2)$-competitive algorithm for $\tau$-HSTs with $\tau \geq \Omega((\log n)^2)$ \cite{BBBT97}.
Since one can assume that $D \leq O(\log n)$ for an $n$-point HST metric (see \cite{BBMN15}),
the mirror descent framework yields an arguably simpler $O((\log n)^2)$-competitive algorithm
for arbitrary HSTs that, moreover,
satisfies the refined guarantees of \pref{thm:refined}.

\paragraph{Unfair metrical task systems}
The algorithms in \cite{BBBT97, FM03} are based on the recursive combination of {\em unfair metrical task systems}, introduced by \cite{Sei99}.
Roughly speaking, one is given an {\em unfairness ratio} $r_x \geq 1$ for every point $x \in X$, and the online algorithm
is charged a service cost of $r_x c_t(x)$ for playing $x \in X$ at time $t$, while the offline algorithm is only charged $c_t(x)$.
Competitive algorithms for unfair task systems are useful in constructing algorithms
for HSTs, where $r_x$ is a proxy for the competitive ratio of an algorithm
on MTS instances defined in a subtree rooted at $x$.

In pursuing this strategy,
Fiat and Mendel \cite{FM03} employ
two different combining
algorithms that can be roughly described as follows:
\begin{enumerate}[label=\textbf{A\arabic*}]
   \item \cite{BBBT97}
      If the unfairness ratios are $\{r_x : x \in X \}$ and $(X,d)$ is an $n$-point uniform metric,
      then one obtains a competitive ratio of $O(\log n) + \max \{ r_x : x \in X \}$. \label{item:alguniform}
   \item \cite{BBBT97,Sei99,BKRS00} If $X=\{x_1,x_2\}$ has $d(x_1,x_2)=1$
      and the unfairness ratios are $r_1, r_2 \geq 1$, then
      one obtains a competitive ratio of
      \[
         r \seteq r_1 + \frac{r_1-r_2}{e^{r_1-r_2}-1}\,.
      \]
      One can observe the following property:  If $r_1 \leq 2 (1+\ln y_1)$ and $r_2 \leq 2(1+\ln y_2)$, then
      $r \leq 2 (1+\ln(y_1+y_2))$. \label{item:algtwopoint}
\end{enumerate}

Our second contribution is to refine this approach using the mirror descent framework.
This allows us to obtain an optimal $O(\log n)$ competitive ratio for MTS
on an arbitrary HST.\footnote{One should recall that this is optimal, up to a universal constant factor,
among all HST metrics, but it is an open problem to establish a lower bound of $\Omega(\log n)$
for {\em every} $n$-point HST metric.}

\begin{theorem} \label{thm:HST}
   There is an $O(\log n)$-competitive randomized
   algorithm for metrical task systems on $n$-point HST metrics.
\end{theorem}

Combined with known HST embedding theorems \cite{Bar96,FRT04},
this yields an $O((\log n)^2)$-competitive randomized algorithm for any $n$-point metric space,
improving the state of the art.

\pref{thm:HST} is proved in \pref{sec:unfair} by presenting a combining algorithm that
correctly
interpolates between the behavior of algorithms \ref{item:alguniform} and \ref{item:algtwopoint} described above:
If $(X,d)$ is a uniform metric with unfairness ratios $\{r_x : x \in X \}$, then 
for some universal constant $C > 1$,
we obtain the ``smooth maximum'' competitive ratio $O(1) \log \left(\sum_{x \in X} \exp(C r_x)\right).$
The combining algorithm fits naturally into the mirror descent framework
by assigning different ``learning rates'' to each piece of the space
based on the corresponding unfairness ratio.

\section{MTS and mirror descent}

We first develop the mirror descent framework
in the context of metrical task systems.
These general principles
will then be applied in \pref{sec:regularization} and \pref{sec:unfair}.

\subsection{Randomized algorithms} Let $\Delta(X)$ be the set of probability measures supported on $X$, and denote by $\vvW_X^1(\mu,\nu)$ the Earthmover distance (a.k.a., the $L^1$ transportation distance) between $\mu, \nu \in \Delta(X)$.  In other words,
$\vvW_X^1(\mu,\nu) = \inf \E d(Y,Z)$, where the infimum is over all random variables $(Y,Z)$ such that $Y$ has law $\mu$ and $Z$ has law $\nu$.

A random state $\rho_t \in X$ is completely specified by its (deterministic) probability distribution $p_t \in \Delta(X)$. Moreover for any deterministic sequence $p_1,\ldots, p_t$ there exists an adapted sequence $\rho_1, \ldots \rho_t$ such that
\[
\E \sum_{t=1}^T d(\rho_{t-1}, \rho_t) =\sum_{t=1}^T  \vvW_X^1(p_{t-1}, p_t) \,,
\]
where $p_0$ is the probability distribution concentrated at $\rho_0 \in X$.
In particular, we see that a randomized online algorithm on $(X,d)$ for the input sequence $\bm{c}$ is equivalently
described by a deterministic online algorithm on the
metric space $(\Delta(X), \vvW_X^1)$ with the cost functions $c_t$ extended linearly from $X$ to $\Delta(X)$.

\subsection{Continuous-time model}
Rather than the discrete time model of the introduction we will work in a continuous model, where $t \in \Z_+$ is replaced by $t \in \R_+$ and discrete sums are replaced by integrals. More precisely, an online algorithm now maps, for any $T \in \R_+$, a continuous path $(c(t))_{t \in [0,T]}$ of cost functions $c(t) : X \rightarrow \R_+$ to a (random) state $\rho(T) \in X$. Denote by
$p(T)$ the law of $\rho(T)$.
The corresponding total expected service cost is defined to be (we omit the dependency on the algorithm and the costs)
\[ 
S \seteq \E \int_{\R_+} c(t)(\rho(t))\,dt = \int_{\R_+} \langle c(t), p(t) \rangle\,dt,
\] and the movement cost is (with the notation $\rho(t^-) = \lim_{s \rightarrow t, s<t} \rho(s)$):
\[
M \seteq \E \sum_{t : \rho(t^-) \neq \rho(t)} d(\rho(t^-), \rho(t)) = \int_{\R_+} \lim_{h \rightarrow 0^-} \frac{\vvW_X^1(p(t+h), p(t))}{|h|} dt.
\]
%Similarly, we use $S^*$ and $M^*$ denote respectively the service cost and movement cost of {\em some} continuous-time offline optimum. 
%(Note that these quantities are not uniquely determined, as the total objective value is $S^*+M^*$, but the decomposition
%into service and movement costs could differ for various optimal offline algorithms.)
The following result is folklore.

\begin{lemma} \label{lem:distocont}
   The existence of an $\alpha$-competitive algorithm for the continuous-time model (with piecewise continuous costs)
   implies the existence of a $\alpha$-competitive algorithm for the discrete-time model.
\end{lemma}

\begin{proof}
Let us describe an update procedure upon receiving a discrete time cost function $C$ in state $\rho_0$. Denote $T \seteq \max_{x \in X} C(x)$. Let $(c(t))_{t \in [0,T]}$ be a ``waterfilling" continuous time version of $C$, that is $c(t)(x) = \1\{C(x) \geq t\}$. Let $(\rho(t))_{t \in [0,T]}$ be the $\alpha$-competitive continuous-time algorithm path on this cost function path, starting from state $\rho(0) = \rho_0$. Let $s \seteq \argmin_{t \in [0,T]} C(\rho(t))$, and notice that by definition of the cost path one has $\int_{0}^T c(t)(\rho(t)) dt \geq C(\rho(s))$.

Furthermore one also has that the movement of the continuous-time algorithm is at least
$d(\rho(0), \rho(s)) + d(\rho(s), \rho(T))$. Thus we see that the discrete-time algorithm can simply update to $\rho(s)$, pay the service cost there, and then move to $\rho(T)$. The total cost of this discrete-time update is smaller than the total cost of the continuous-time update, and furthermore both algorithms end up in the same state so that one can repeat the argument for the next discrete-time cost function. On the other hand, the cost of the offline optimum in the continuous-time model with cost $c$ is clearly smaller than in the discrete-time model with cost $C$ (simply because the continuous-time cost of the offline discrete time optimal path is equal to its discrete-time cost). This concludes the proof.
\end{proof}

Note that, in fact, the above proof shows that one can preserve the {\em refined guarantees}
(see \pref{thm:refined}) from the continuous-time model to the discrete-time model.

\subsection{State representation}

For $x = (x_1,\hdots,x_N)$ with $N \geq n$, denote $P_n(x) = (x_1,\hdots, x_n)$. Let $\mathsf{K}  \subset \R^N$ be a convex body such that
\begin{equation} \label{eq:slice}
\left\{P_n(x) : x \in \mathsf{K}  \right\} = \left\{p \in \R_+^n : \sum_{i=1}^n p_i =1 \right\}.
\end{equation}
We will associate the latter set with $\Delta(X)$ by taking
$X = \{1,2,\ldots,n\}$.

We will also assume that there exists a norm $\|\cdot\|$ on $\R^N$ such that
\begin{equation}\label{eq:Wnorm}
\vvW_X^1(P_n(x), P_n(y)) = \|x-y\| \,.
\end{equation}
This assumption is specific to the setting of tree metrics.

Thus instead of a randomized online algorithm on $(X,d)$ against costs functions $c_t~:~X~\rightarrow~\R_+$, we will specify a deterministic online algorithm on $(\mathsf{K},\|\cdot\|)$ against linear cost functions $(c_t(1),\hdots, c_t(n), 0, \hdots) \in \R_+^N$. (With a slight abuse of notation, we will use $c_t$ for this cost function.)

\subsection{Mirror descent dynamics} Let $\Phi : \mathsf{K} \to \R$
be a strictly convex function. Denote by $N_{\mathsf{K}}(x) = \{\theta \in \R^N : \theta \cdot (y-x) \leq 0, \ \forall y \in \mathsf{K}\}$ 
the normal cone of $\mathsf{K}$ at $x$.
In our recent work \cite{BCLLM18} on the $k$-server problem, we considered the following dynamics to respond to a continuous time linear cost $(c(t))_{t \in \R_+}$ starting in some state $x_0 \in \mathsf{K}$:
\begin{align}
   \nabla^2 \Phi(x(t)) x'(t) &= - (c(t) + \lambda(t)), \, \lambda(t) \in N_\mathsf{K}(x(t))\label{eq:inclusion}\\
   x(0) &= x_0\nonumber \\
   x(t) &\in \mathsf{K} \qquad \forall t \geq 0\,.\nonumber
\end{align}
\cite[Thm. 2.1]{BCLLM18} shows that under mild regularity assumptions (which will be satisfied here), the above differential inclusion admits a unique and absolutely continuous solution. Absolute continuity implies
(see, e.g., \cite[Lem. 3.45]{Leo09}) that for almost every $t \in \R_+$,
\begin{equation} \label{itwastrue}
x_i(t) = 0 \Rightarrow x_i'(t) = 0 \,.
\end{equation}
Futhermore, if $\mathsf{K}$ is a polyhedron given by
\begin{equation}\label{eq:polyhedron}
   \mathsf{K} = \left\{ x \in \R^N : A x \leq b \right \}, \qquad A \in \R^{m \times N}, \ b \in \R^m,
\end{equation}
then, there is a measurable $\hat{\lambda} : [0,\infty) \to \R^m$ such that $\hat{\lambda}$ represents the normal force $\lambda$
\begin{equation}\label{eq:multiplier}
   A^\top \hat{\lambda}(t) = \lambda(t)\,,\quad t \geq 0\,
\end{equation}
and $\hat{\lambda}$ satisfies the
complementary-slackness conditions:
For all $i=1,2,\ldots,m$ and almost all $t \geq 0$:
\begin{equation}\label{eq:cslack}
   \hat{\lambda}_i(t) > 0 \implies \langle A_i, x(t)\rangle = b_i\,,
\end{equation}
where $A_i$ is the $i$th row of $A$.
We will fix such a representation $\hat{\lambda}$
and call $\hat{\lambda}_i$ the Lagrangian multiplier
of the constraint $\langle A_i, x\rangle \leq b_i$.

\subsection{Cost of the algorithm}

Recall that the Bregman divergence associated to $\Phi$ is defined by
\[
   {D}_{\Phi}(y ; x) \seteq \Phi(y) - \Phi(x) - \langle \nabla \Phi(x), y -x \rangle \geq 0\,,
\]
where the latter inequality follows from convexity of $\Phi$.
Finally, we
denote 
\[
   \mathrm{Lip}_{\|\cdot\|}(\Phi) \seteq \sup_{x, y \in \mathsf{K}} \|\nabla \Phi(x) - \nabla \Phi(y)\|^*\,,
\]
where $\|\cdot\|^*$ is the dual norm to $\|\cdot\|$ on $\R^N$.

\begin{lemma} \label{lem:basic}
The mirror descent path \eqref{eq:inclusion} satisfies, for any absolutely
continuous comparator path $(y(t))_{t \geq 0}$ in $\mathsf{K}$,
\[
\int_{\R_+} \langle c(t), x(t) \rangle\,dt \leq \int_{\R_+} \langle c(t), y(t)\rangle\, dt + \mathrm{Lip}_{\|\cdot\|}(\Phi) \int_{\R_+} \|y'(t)\|\,dt
+ \mathrm{Lip}_{\|\cdot\|}(\Phi) \cdot \|y(0) - x(0)\| \,.
\]
\end{lemma}

\begin{proof}
For any fixed $y \in \mathsf{K}$ one has:
\begin{align}
   \partial_t {D}_{\Phi}(y ; x(t)) &= - \langle \nabla^2 \Phi(x(t)) x'(t), y - x(t)\rangle \label{eq:bregman-deriv}\\
& = \langle c(t) + \lambda(t), y - x(t)\rangle \nonumber \\
& \leq  \langle c(t), y - x(t)\rangle\,, \nonumber
\end{align}
where the inequality follows from $\lambda(t) \in N_\mathsf{K}(x(t))$ and $y \in \mathsf{K}$.
Furthermore, for any fixed $x \in \mathsf{K}$, one has:
\begin{align}
\partial_t {D}_{\Phi}(y(t) ; x) & = \llangle \nabla \Phi(y(t)) - \nabla \Phi(x)), y'(t)\rrangle \nonumber \\
                                & \leq \mathrm{Lip}_{\|\cdot\|}(\Phi) \cdot \|y'(t)\| \,. \label{eq:lipbound}
\end{align}

Combining both inequalities, for any time $T$, we have
\begin{align}
D_{\Phi}(y(T);x(T))-D_{\Phi}(y(0);x(0))&=\int_{0}^{T}\partial_{x(t)}D_{\Phi}(y(t);x(t))+\partial_{y(t)}D_{\Phi}(y(t);x(t))\,dt \nonumber\\
&\leq\int_{0}^{T} \langle c(t), y(t)-x(t)\rangle +\mathrm{Lip}_{\|\cdot\|}(\Phi)\cdot \|y'(t)\|\,dt. \label{eq:basic1}
\end{align}
To bound the left hand side, we note that by convexity of $\Phi$,
\begin{eqnarray*}
   D_{\Phi}(y(0);x(0)) &	\leq & \langle \nabla \Phi(y(0)) - \nabla \Phi(x(0)) , y(0) - x(0) \rangle \\
   %\left|\Phi(y(0))-\Phi(x(0))\right|+\left|\langle \nabla\Phi(x(0)),y(0)-x(0)\rangle\right| \\
	& \leq & \mathrm{Lip}_{\|\cdot\|}(\Phi)\cdot\|y(0)-x(0)\|\,,
\end{eqnarray*}
and $D_{\Phi}(y(T);x(T))\geq0$. Putting these into  (\ref{eq:basic1}) gives that
\begin{eqnarray*}
-  \mathrm{Lip}_{\|\cdot\|}(\Phi)\times\|y(0)-x(0)\|  \leq\int_{0}^{T} \langle c(t), y(t)-x(t)\rangle+\mathrm{Lip}_{\|\cdot\|}(\Phi)\cdot\|y'(t)\|\,dt.
\end{eqnarray*}
The result follows by taking $T \to +\infty$.
\end{proof}

\subsection{Reduced costs}

Consider a polyhedron
of the form $\mathsf{K} = \{x \in \R^N : A x \leq b \, \text{and} \, x_i \geq 0, \forall i \in [n]\}$. Then \eqref{eq:multiplier}
shows that the normal force is given by $\lambda(t) = A^\top \hat{\lambda}(t) - \xi(t)$ where $\xi(t) > 0$ is a Lagrange multiplier of the constraints $x_i \geq 0, i \in [n]$. We refer to the quantity
$c(t) - \xi(t)$ as a {\em reduced cost}.
Intuitively, the reduced cost is the ``effective'' cost for the algorithm's dynamics:
Observe that $\langle x(t), c(t)-\xi(t)\rangle = \langle x(t),c(t)\rangle$ since $\xi_i(t) > 0 \implies x_i(t)=0$.
Although reduced costs are not unique, the following lemma shows that any reduced cost is bounded in a sense we now describe.
Certain greek letters (e.g., $\beta,\eta,\delta$) will sometimes represent scalar parameters,
and sometimes vectors of parameters.  In the latter case, we use the bold versions (e.g., $\bm{\beta},\bm{\eta},\bm{\delta}$).

\begin{lemma} \label{lem:reas}
   Assume that $\mathsf{K} = \{x \in \R_+^N : A x \leq b\}$, for some $A \in \R^{m \times N}, b \in \R^m$. Let  $x : \R_+ \rightarrow \mathsf{K}$ be the mirror descent path \eqref{eq:inclusion} for some regularizer $\Phi$. Let also $\bm{\delta} \in \mathsf{K}$ be such that $A \bm{\delta} = b$. 
Then one has:
\[
   \int_{\R_+} \langle c(t)-\xi(t),\bm{\delta}\rangle\, dt \leq \int_{\R_+} \langle c(t), x(t)\rangle\, dt + \mathrm{Lip}_{\|\cdot\|}(\Phi) \cdot \sup_{x \in \mathsf{K}} \|x - \bm{\delta}\|\,,
\]
where $\xi(t)$ is any Lagrange multiplier for the constraint $x\geq 0$.
\end{lemma}

\begin{proof}
Note that that $\lambda(t) =  A^{\top} \hat{\lambda}(t)- \xi(t)$ for some $\hat{\lambda}(t) \in \R_+^m$ and $\xi(t) \in \R_+^N$. 
In particular, one has:
\[
   \langle \lambda(t) +\xi(t), \bm{\delta} - x(t)\rangle  = \langle \hat{\lambda}(t), A(\bm{\delta} - x(t))\rangle \geq 0 \,,
\]
where the inequality uses the fact that $A \bm{\delta} = b \geq A x(t)$. Thus one obtains:
\begin{align*}
   \llangle c(t) -\xi(t), \bm{\delta} - x(t)\rrangle &\leq \llangle c(t) + \lambda(t),  \bm{\delta} - x(t)\rrangle \\
                                                     &= \llangle \nabla^2 \Phi(x(t)) x'(t), x(t) - \bm{\delta}\rrangle \\
                                                     &\stackrel{\mathclap{\eqref{eq:bregman-deriv}}}{=} \partial_t D_{\Phi}(\bm{\delta}; x(t))\,.
%                                                     &\stackrel{\mathclap{\eqref{eq:lipbound}}}{\leq} \Lip_{\|\cdot\|}(\Phi) \cdot \sup_{x \in \mathsf{K}} \|x-\bm{\delta}\|\,.
\end{align*}
Using that $\xi_i(t) \neq 0 \implies x_i(t) = 0$ for almost every $t \geq 0$ (recall \eqref{eq:cslack}),
integrating over $\R_+$ yields
\[
   \int_{\R_+} \langle c(t)-\xi(t), \bm{\delta}\rangle \,dt \leq \int_{\R_+} \langle c(t),x(t)\rangle\,dt + 
   \sup_{x \in \mathsf{K}} D_{\Phi}(\bm{\delta}; x)\,,
\]
which yields the desired result.
\end{proof}

\section{Entropic regularization}
\label{sec:regularization}

We will now instanatiate our regularizer $\Phi$ to be an appropriate
weighted and shifted entropy.
We first apply this to weighted star metrics, and then to general (bounded depth) trees.

\subsection{Warm-up: Weighted stars} \label{sec:warmup}

We consider here the case where $X$ is the set of leaves in a weighted star.
Let $w_i > 0$ be the weight on the edge from the $i^{th}$ leaf to the root, and denote $\Delta \seteq \max_{i \in [n]} w_i$.
We set $\mathsf{K} \seteq \{x \in \R_+^n : \sum_{i=1}^n x_i = 1\}$, and
the norm measuring movement in \eqref{eq:Wnorm} is the weighted $\ell_1$ norm on $\R^n$
given by $\|\xi\| \seteq \sum_{i=1}^n w_i |\xi_i|$ (note that the dual norm is an inversely weighted $\ell_{\infty}$ norm, namely $\|g\|^* \seteq \max_{i \in [n]} \frac{|g_i|}{w_i}$).

We use the regularizer $\Phi(x) \seteq \frac{1}{\eta} \sum_{i=1}^n w_i (x_i + \delta) \log(x_i + \delta)$ where $\eta>0$ is a learning rate and $\delta \in [0,1/2]$ is a shift parameter.
Now \eqref{eq:inclusion} yields the following dynamics:
\begin{equation} \label{eq:dynstar}
x_i'(t) = \frac{\eta}{w_i} (x_i(t) + \delta) (\mu(t) - c_i(t) + \xi_i(t))\qquad i=1,2,\ldots,n \,,
\end{equation}
where $\mu(t) \in \R$ is a Lagrange multiplier corresponding to the constraint $\sum_{i=1}^n x_i(t) = 1$, and $\xi_i(t) \geq 0$ is a Lagrange multiplier corresponding to the constraint $x_i(t) \geq 0$. 

\begin{theorem} \label{thm:star}
If there is an offline algorithm with service cost $S^*$ and movement cost $M^*$, then
the above algorithm satisfies 
\[
S \leq S^* + \frac{2 \log(1/\delta)}{\eta} M^*\,,
\]
and
\begin{equation}\label{eq:movecost-stars}
M \leq 2 \eta (1+\delta n) S + \left(1+8 \delta n \log(1/\delta)\right) \Delta \,.
\end{equation}
Taking $\eta = 4 \log n$ and $\delta = 1/n^2$ thus yields a $O(\log n)$-competitive algorithm (with $1$-competitive service cost in the sense that $S \leq S^* + M^*$).
\end{theorem}

\begin{proof}
First notice that $(\nabla \Phi(x) - \nabla \Phi(y))_i = \frac{w_i}{\eta} (\log(x_i + \delta) - \log(y_i + \delta))$, and thus 
\[
\mathrm{Lip}_{\|\cdot\|}(\Phi) = \sup_{x,y \in K, i \in [n]} \frac{|(\nabla \Phi(x) - \nabla \Phi(y))_i|}{w_i} \leq \frac{2 \log(1/\delta)}{\eta} \,.
\]
Let $(y(t))_{t \geq 0}$ denote the path of some piecewise-continuous offline algorithm
achieving $S^*$ and $M^*$.
Applying \pref{lem:basic} yields
\begin{align*}
S=\int_{\R_{+}} \langle c(t), x(t)\rangle\,dt & \leq\int_{\R_{+}}\langle c(t), y(t)\rangle\,dt+\mathrm{Lip}_{\|\cdot\|}(\Phi)\int_{\R_{+}}\|y'(t)\|\,dt+ \mathrm{Lip}_{\|\cdot\|}(\Phi)\cdot \|y(0)-x(0)\|\\
 & =S^{*}+\frac{2\log(1/\delta)}{\eta}M^{*}\,,
\end{align*}
where the final equality uses $x(0)=y(0)$.

Before bounding the movement, we observe that $\mu(t) \geq 0$ almost surely. This follows from $c_i(t)\geq 0$ together with the following identity
\[
0 = \sum_{i: x_i(t)\neq0} x'_i(t)  =   \sum_{i: x_i(t)\neq0} \frac{\eta}{w_i} (x_i(t) + \delta) (\mu(t) - c_i(t) ) \,.\]
where the first equality holds almost surely (by \eqref{itwastrue}) from differentiating $\sum_{i=1}^n x_i(t) = 1$, and second follows from complementary slackness \eqref{eq:cslack}
for $\xi(t)$.
Moreover one also has $c_i(t) \geq \xi_i(t)$ almost surely. Indeed, again by \eqref{itwastrue} and complementary slackness, one has $\xi_i(t) > 0 \Rightarrow x(t)=0 \Rightarrow x_i'(t) = 0 \Rightarrow \mu + \xi_i(t) - c_i(t) = 0$, which shows that $c_i(t) - \xi_i(t) \geq 0$ (since $\mu(t) \geq 0$).

For the movement, we first note that (with the notation $(x)_- \seteq (x_1 1\{x_1 <0\}, \hdots, x_n 1\{x_n<0\})$):
\begin{align}
   M \leq 2\int_{\R_+} \|(x'(t))_-\|\,dt + \Delta. \label{eq:star1}
\end{align}
To calculate $\|(x'(t))_-\|$, note that the dynamics \eqref{eq:dynstar} and $\mu(t) \geq 0$ give $x_i'(t) \geq - \frac{\eta}{w_i} (x_i(t) + \delta) (c_i(t) - \xi_i(t))$, and furthermore since  $c(t) \geq \xi(t)$ one also has $|(x_i'(t))_-| \leq \frac{\eta}{w_i} (x_i(t) + \delta) (c_i(t) - \xi_i(t))$, which yields:
\[
\sum_{i : x_i'(t) \leq 0} w_i |x_i'(t)| \leq \eta \llangle x(t) +\delta \1, (c(t) - \xi(t)\rrangle \,.
\]
Hence (\ref{eq:star1}) gives
\begin{align*}
   M \leq2\eta\left(\int_{\R_+} \langle x(t)+\delta\1, c(t)-\xi(t)\rangle\,dt\right)+\Delta =2\eta S+2\eta\delta\int_{\R_+} 
   \llangle \1,c(t)-\xi(t)\rrangle\,dt +\Delta
\end{align*}
where we used that $\xi_i(t) = 0$ if $x_i(t) >0$.

Now, an application of 
\pref{lem:reas} shows that $$\int_{\R_+} \llangle c(t)-\xi(t), \frac{1}{n}\1\rrangle\,dt\leq\int_{\R_+} \llangle c(t), x(t)\rrangle\,dt+\frac{4\log(1/\delta)}{\eta}\Delta$$ which our verification of \eqref{eq:movecost-stars}.
\end{proof}

\newcommand{\vvT}{\vvmathbb{T}}
\newcommand{\vvM}{\vvmathbb{M}}
\newcommand{\ch}{\chi}
\newcommand{\extmeas}{\widehat{M}}
\renewcommand{\root}{\vvmathbb{r}}

\subsection{The multiscale entropy and MTS on trees}

Consider now a rooted tree $\cT = (V,E)$ with root $\root \in V$
and leaves $\cL \subseteq V$.
Let $\{w_v > 0 : v\in V  \textbackslash \{\root\}\}$ be a collection of positive
weights on $ V  \textbackslash \{\root\}$ (except that $w_{\root} = 0$).
We will assume (without loss of generality)
that every leaf $\ell \in \cL$ is at the same
combinatorial distance $D$ from the root. 
For $u \in V \setminus \{\root\}$, let $p(u) \in V$ denote the parent of $u$.
Let $\dist_w(x,y)$ denote the weighted path distance between $x,y \in V$,
where an edge $\{p(u), u\}$ is given weight $w_u$.

Our setting is now $(X,d) \seteq (\cL, \dist_w)$.
The natural norm in which to measure movement (recall \eqref{eq:Wnorm}) is the weighted $\ell_1$ norm on an expanded state space: 
For $z \in \R^V$, we denote
\[
   \|z\| \seteq \left\|z\right\|_{\ell_1(w)} = \sum_{v \in V} w_v \left|z_v\right|\,,
\]
and we set 
\[
\mathsf{K} \seteq \left\{x \in \R^V : x_{\root} = 1, \, \text{and} \, \forall u \in V \setminus \cL , x_u \leq \sum_{v : p(v) = u} x_v , \, \text{and} \, \forall \ell \in \cL, x_{\ell} \geq 0 \right\} \,.
\] 
%We note that $\mathsf{K}$ does not exactly satisfy \eqref{eq:slice}, in the sense that the appropriate slice contains {\em supermeasures} of total mass greater than $1$.  However the mirror descent dynamics will maintain a proper measure in that slice, which is sufficient, see \pref{lem:stayameasure}.
We note that $\mathsf{K}$ does not enforce that 
the total mass of each slice (all vertices at the same
height) is exactly one, nor does it enforce that all variables are nonnegative.
However the mirror descent dynamics will implicitly maintain these constraints.

\paragraph{Mirror descent dynamics}
Define
\[
   \Phi(x) \seteq \frac{1}{\eta} \sum_{u \in V} w_u (x_u + \delta_u) \log(x_u + \delta_u)\,,
\]
where $\eta>0$ is a learning rate and $\bm{\delta} \in (0,1]^V$ is a shift parameter satisfying $\delta_{\root}=1$ and
\[
  \delta_u = \sum_{v : p(v) = u} \delta_v \qquad \forall u \in V\,.
\]
Note that $\Phi$ is well-defined in a neighborhood of the positive orthant $\R_+^V$, which will be sufficient to make the dynamics on $\mathsf{K}$ well-defined (see \pref{lem:stayameasure}).

Using the formula for the normal cone of $\mathsf{K}$, this gives the following dynamics,
where we use the notation $\odot$ for the Hadamard (entrywise) product:
\begin{equation} \label{eq:dyntrees}
   w \odot x'(t) = - \eta (x(t) + \bm{\delta}) \odot (c(t) + \lambda(t) - \xi(t) - \mu(t)) \,,
\end{equation}
where $\lambda(t) = \sum_{u \in V} \hat{\lambda}_{u}(t) \left(e_u - \sum_{v : p(v) = u} e_v\right)$ for some $\hat{\lambda}_{u}(t) \geq 0$, $\xi(t) = \sum_{\ell \in \cL} \xi_{\ell}(t) e_{\ell}$ for some $\xi_{\ell}(t) \geq 0$, and $\mu(t) = \hat{\mu}(t) e_{\root}$ for some $\hat{\mu}(t) \geq 0$.
%Note that $\lambda(t)$ can be both positive and negative.

\begin{lemma} \label{lem:stayameasure}
For almost all $t \geq 0$, $\sum_{u \in \cL} x_u(t) = 1$, and $x_u(t) \geq 0$ for all $u \in V$.
\end{lemma}
\begin{proof}
We prove that for almost all $t \geq 0$, $x_u(t) \geq \sum_{v:p(v)=u}x_{v}(t)$ for all $u \notin \cL$.
Since $x(t) \in \mathsf{K}$, this implies that the inequality holds with equality,
and suffices to establish the lemma (due to $x_{\root}=1$ and $\{x_{\ell} \geq 0: \ell \in \cL\}$).

For any $u\notin\cL$, denote
\[
   I_u \seteq \left\{t \geq 0 : x_{u}(t) < \sum_{v:p(v)=u}x_{v}(t)\right\}.
\]
For any $t\in I_u$, we have $\hat{\lambda}_{u}(t)=0$ by complementary slackness, and $c_{u}(t)=0$ since $u\notin\cL$, which together imply $x_{u}'(t)\geq0$. Now for a child $v$ of $u$, it easy to see that $\hat{\lambda}_u(t)=0$ implies that if $x_v(t) \geq 0$ then $x_v'(t) \leq 0$ (indeed, $\hat{\lambda}_u(t)$ is the only part in \eqref{eq:dyntrees} which can induce a strictly positive value for $x'_v(t)$). 
%In particular, \eqref{itwastrue} allows to remove the condition $x_v(t) \neq 0$ in the implication, yielding $x_{v}'(t)\leq0$ for almost all $t \in I_u$.
Let us assume by induction on the level of the tree (starting from the leafs) that we have already proved $x_v(t) \geq 0$. 
Thus we have $x'_{u}(t)\geq\sum_{v:p(v)=u}x_{v}'(t)$ for almost all $t$ such that $x_{u}(t)<\sum_{v:p(v)=u}x_{v}(t)$, implying that $I_u$ has measure zero for each $u \notin \cL$, and hence concluding the proof.
\end{proof}

\begin{theorem}
If there is an offline algorithm with service cost $S^*$ and movement cost $M^*$, then
the above algorithm satisfies
\[
S \leq S^* + \frac{2 \log(1/\min_{u \in \cL} \delta_u)}{\eta} M^* \,,
\]
and
\begin{equation}\label{eq:movecost-trees}
M \leq 4 \eta D S + \left(1+ 2 D + 8  D\log(1/\min_{u \in \cL} \delta_u)\right) \mathrm{diam}(X) \,.
\end{equation}
Taking $\eta \seteq 2\log n$ and $\delta_u \seteq 1/n$ for $u \in \cL$ thus yields an  $O(D \log n)$-competitive algorithm with
$1$-competitive service cost in the sense that $S \leq S^* + M^*$.
\end{theorem}

\begin{proof}
First notice that $(\nabla \Phi(x) - \nabla \Phi(y))_u = \frac{w_u}{\eta} (\log(x_u + \delta_u) - \log(y_u + \delta_u))$, and thus 
\[
\mathrm{Lip}_{\|\cdot\|}(\Phi) = \sup_{x,y \in K, u \in V} \frac{|(\nabla \Phi(x) - \nabla \Phi(y))_u|}{w_u} \leq \frac{2 \log(1/ \min_{u \in \cL} \delta_u)}{\eta} \,.
\]
Let $(y(t))_{t \geq 0}$ denote the path of some piecewise-continuous offline algorithm
achieving $S^*$ and $M^*$.  Then \pref{lem:basic} yields
\begin{align*}
S=\int_{\R_{+}} \langle c(t), x(t)\rangle\,dt & \leq\int_{\R_{+}}\langle c(t), y(t)\rangle\, dt+\mathrm{Lip}_{\|\cdot\|}(\Phi)\int_{\R_{+}}\|y'(t)\|\,dt+2\mathrm{Lip}_{\|\cdot\|}(\Phi)\cdot \|y(0)-x(0)\|\\
 & =S^{*}+ \frac{2 \log(1/ \min_{u \in \cL} \delta_u)}{\eta}M^{*}.
\end{align*}
%First notice that $\mathrm{Lip}_{\|\cdot\|}(\Phi) \leq \frac{2 \log(1/ \min_{u \in \cL} \delta_u)}{\eta}$, and thus the first inequality follows from \pref{lem:basic}.

For the movement, we note that
\begin{align}
   M \leq 2\int_{\R_+} \|(x'(t))_+\|\,dt +  \mathrm{diam}(X). \label{eq:tree1}
\end{align}
To calculate $\|(x'(t))_+\|$, note that the dynamics \eqref{eq:dyntrees} gives for almost all $t$ (recall that by \eqref{itwastrue}, $x_u'(t) > 0 \Rightarrow \xi_u(t) = 0$):
%For the movement note that up to a multiplicative factor $2$ and additive $\mathrm{diam}(X)$ it suffices to calculate $\int \|(x'(t))_+\|_{\ell_1(w)}$, and in particular the dynamics \eqref{eq:dyntrees} gives:
\[
\sum_{u : x_u'(t) > 0} w_u x_u'(t) \leq \eta \sum_{u \in V \setminus \cL} \hat{\lambda}_u(t) \sum_{v : p(v) = u} (x_v(t) + \delta_v) = \eta \sum_{u \in V \setminus \cL} \hat{\lambda}_u(t) (x_u(t) + \delta_u) \,,
\]
where the equality uses the assumption on $\bm{\delta}$ and the fact that $\hat{\lambda}_u(t) \neq 0 \Rightarrow x_u(t) = \sum_{v : p(v) = u} x_v(t)$.

In analogy with the auxiliary depth potential employed in \cite{BCLLM18},
we consider the weighted depth:
\[
   \Psi(x) \seteq \sum_{u \in V} d_u w_u x_u,
\]
where $d_u$ is the combinatorial depth of $u$.

Using $c_u(t) = \xi_u(t) = 0$ for $u \not\in \cL$, we have
\begin{eqnarray*}
   \partial_t \Psi(x(t)) & = & - \eta D \llangle c(t) -\xi(t), x(t) + \bm{\delta}\rrangle - \eta  \sum_{u \in V \setminus \cL} \hat{\lambda}_{u}(t) \left(d_u (x_u(t) + \delta_u) - \sum_{v : p(v) = u} d_v(x_v(t) + \delta_v)\right) \\
                         & = & - \eta  D\llangle c(t) -\xi(t), x(t) + \bm{\delta}\rrangle + \eta  \sum_{u \in V \setminus \cL} \hat{\lambda}_{u} (t) (x_u(t) + \delta_u) \,,
\end{eqnarray*}
where the second equality uses $\sum_{v : p(v) = u} d_v(x_v(t) + \delta_v) = (d_u+1) \sum_{v : p(v) = u} (x_v(t) + \delta_v) = (d_u+1) (x_u(t) + \delta_u)$.
Combining the two above displays one obtains
\[
   \sum_{u : x_u'(t) \geq 0} w_u x_u'(t) \leq \partial_t \Psi(x(t)) + \eta D \llangle c(t)-\xi(t), x(t) + \bm{\delta}\rrangle \,.
\]
Putting it into (\ref{eq:tree1}) gives
\begin{align*}
   M & \leq2(\Psi(x(T))-\Psi(x(0)))+2\eta DS+2\eta D \int_{\R_+} \llangle \bm{\delta}, c(t)-\xi(t)\rrangle\,dt+\mathrm{diam}(X)\\
     & \leq2D\cdot\mathrm{diam}(X)+2\eta DS+2\eta D \int_{\R_+} \llangle \bm{\delta}, c(t)-\xi(t)\rrangle\,dt+\mathrm{diam}(X)\,.
\end{align*}

By \pref{lem:reas}, we have
$$\int_{\R_+} \llangle \bm{\delta}, c(t)-\xi(t)\rrangle\, dt \leq \int_{\R_+} \llangle c(t), x(t)\rrangle\, dt + \frac{4\log(1/ \min_{u \in \cL} \delta_u)}{\eta}  \mathrm{diam}(X)$$ which establishes \eqref{eq:movecost-trees} and completes the proof.
%The proof of the second inequality in the theorem statement is now completed by invoking \pref{lem:reas} (observe that again due to the dynamics one can clearly assume without loss of generality that the cost is reasonable).
\end{proof}

\section{Unfair MTS and subspace gluing}
\label{sec:unfair}

We now apply the mirror descent framework to the unfair MTS problem,
yielding an optimal gluing strategy for HSTs.

\subsection{Log-sum-exp gluing on a weighted star}
Let $\bm{\beta} \in \R_+^n$ and $\gamma \in [1,+\infty)$ be {\em unfairness ratios}.
The unfair service cost to service the cost vector $c \in \R_+^n$ in a state $x \in \R_+^n$ is defined to be
$\langle \bm{\beta} \odot c, x\rangle$, while the unfair movement cost is the movement cost multiplied by $\gamma$. 

In unfair MTS, the online algorithm's total cost is the sum of the unfair service cost $S^u$
and unfair movement cost $M^u$, given by
\begin{align*}
   S^u &\seteq \int_{\R_+} \llangle \bm{\beta} \odot c(t), x(t)\rrangle\, dt \\
   M^u &\seteq \gamma \int_{\R_+} \|x'(t)\|\,dt,
\end{align*}
while an offline algorithm $(y(t))_{t \in \R_+}$ is still evaluated through the sum of (regular) service cost $S^* = \int_{\R_+} \langle c(t), y(t)\rangle\, dt$ and (regular) movement cost $M^* = \int_{\R_+} \|y'(t)\|\,dt$.

We will now consider unfair MTS on weighted stars (recall \pref{sec:warmup}), an
thus we define $\mathsf{K} \seteq \{x \in \R_+^n : \sum_{i=1}^n x_i = 1\}$ and $\|x\| \seteq \sum_{i=1}^n w_i |x_i|$.
Writing $\beta_i = \log u_i$, one can see from the refined guarantees of \pref{thm:star} that a competitive ratio of $O(\gamma \log n + \max_{i \in [n]} \log u_i)$ is achievable for unfair MTS.
We will see now that
one can obtain a competitive ratio of order $O(\gamma \log\left(\sum_{i=1}^n u_i\right))$. 

\paragraph{Entropic regularization with multiple learning rates}
We will use the regularizer
\[
   \Phi(x) \seteq \sum_{i=1}^n \frac{w_i}{\eta_i} (x_i + \delta_i) \log(x_i + \delta_i)\,,
\]
where $\bm{\eta} \in (0,\infty)^n$ is a set of learning rates and $\bm{\delta} \in (0,1/2]^n$ are shift parameters.
This gives the following dynamics: 
\begin{equation} \label{eq:dynstar2}
x_i'(t) = \frac{\eta_i}{w_i} (x_i(t) + \delta_i) (\mu(t) - c_i(t) + \xi_i(t))\,,
\end{equation}
where $\mu(t) \in \R^n$ is a Lagrange multiplier corresponding to the constraint $\sum_{i=1}^n x_i(t) = 1$, and $\xi_i(t)$ is a Lagrange multiplier corresponding to $x_i(t) \geq 0$.

\paragraph{Unfair cost and fair cost} The mirror descent analysis (\pref{lem:basic}) naturally tracks the fair service cost. In order to get an estimate on the unfair service cost, we propose to use multiple learning rates so that the sum of the unfair service cost and movement cost is proportional to the fair service cost. Indeed, \eqref{eq:dynstar2} gives
\[
   \|(x'(t))_-\| \leq \llangle \bm{\eta} \odot \left(c(t)-\xi(t)\right), x(t)+\bm{\delta}\rrangle,
\]
and thus
\begin{align*}
   2 \gamma \|(x'(t))_-\| + \llangle \bm{\beta} \odot c(t), x(t)\rrangle \leq \llangle (\bm{\beta} + 2\gamma \bm{\eta}) \odot c(t), x(t)\rrangle + \llangle c(t)-\xi(t), 2\gamma \bm{\eta} \odot \bm{\delta}\rrangle \,
\end{align*}
We now naturally pick $\bm{\eta}$ such that $\bm{\beta}+2\gamma\bm{\eta} = \zeta \1$ for some constant $\zeta \geq 0$.
Employing \pref{lem:basic} and \pref{lem:reas}, one obtains the following.

\begin{theorem} \label{thm:unfair}
   With $\bm{\eta}$ such that $\beta_i+2 \gamma \eta_i = \zeta$ for all $i$ and $\gamma \geq 1$, the algorithm (\ref{eq:dynstar2}) satisfies
   \begin{equation}\label{eq:unfair-gen}
   S^u + M^u \leq \left(\zeta + 2\gamma \langle \bm{\eta}, \bm{\delta}\rangle\right) \left(S^* +  L M^*\right) + \left(1 + 2\zeta L+6 \gamma L\langle \bm{\eta},
   \bm{\delta}\rangle\right)  \Delta \,,
\end{equation}
where $L \seteq \max_{i \in [n]} \frac{2\log(1/\delta_i)}{\eta_i}$.

Assume now that $\beta_i = 8 \gamma (\log(u_i) + C)$ with $u_i > 0$ and $C\geq0$. Taking $\eta_i = 4 \log(U/u_i)$ and $\delta_i = (u_i/U)^2$ with
$U \seteq \sum_{i=1}^n u_i$, yields
an $8 \gamma (\log(U)+ C+ 1)$-competitive algorithm for unfair MTS.
More precisely we have the following inequality, even when the offline algorithm
is allowed to start in a different state than the online algorithm:
\begin{equation} \label{eq:tobeusedincombi}
S^u + M^u \leq 8 \gamma (\log(U)+ C+ 1) (S^* + M^* + 4 \Delta) \,.
\end{equation}
\end{theorem}

\begin{proof} Notice that
\begin{align}
S^{u}+M^{u} & =\int_{\R_{+}} \llangle \bm{\beta}\odot c(t), x(t)\rrangle\, dt+\gamma\int_{\R_{+}}\|x'(t)\|\,dt \nonumber \\
 & \leq\int_{\R_{+}} \llangle \bm{\beta}\odot c(t), x(t)\rrangle\,dt+2\gamma\int_{\R_{+}}\|(x'(t))_{-}\|\,dt+\Delta\,. \label{eq:unfair_zeta}
\end{align}

Before bounding the movement, we observe that  $\mu(t) \geq 0$. This follows from the equality
$$0 = \sum_{i: x_i(t)\neq0} x'_i(t)  =   \sum_{i: x_i(t)\neq0} \frac{\eta_i}{w_i} (x_i(t) + \delta_i) (\mu(t) - c_i(t) )$$
and the fact that $c_i(t)\geq 0$.

To calculate $\|(x'(t))_-\|$, note that the dynamics \eqref{eq:dynstar2} and $\mu(t) \geq 0$ give
\begin{align*}
   \sum_{i : x_i'(t) \leq 0} w_i |x_i'(t)| & \leq \llangle \bm{\eta} \odot (x(t) + \bm{\delta}), c(t) - \xi(t)\rrangle \\
                                           & = \llangle \bm{\eta} \odot  x(t), c(t)\rrangle  + \llangle \bm{\eta} \odot \bm{\delta}, c(t) - \xi(t)\rrangle.
\end{align*}
Hence (\ref{eq:unfair_zeta}) yields
\begin{align*}
S^{u}+M^{u} & \leq\int_{\R_{+}}\llangle (\bm{\beta}+2\gamma\bm{\eta})\odot c(t)),  x(t)\rrangle\,dt+
\int_{\R_{+}} \llangle c(t)-\xi(t), 2\gamma\bm{\eta}\odot\bm{\delta}\rrangle\,dt+\Delta\\
& =\zeta S+\int_{\R_{+}} \llangle c(t)-\xi(t),2\gamma\bm{\eta}\odot\bm{\delta}\rrangle\,dt+\Delta.
\end{align*}
Now, notice that $\mathrm{Lip}_{\|\cdot\|}(\Phi) \leq \max_{i} \frac{2 \log(1/\delta_i)}{\eta_i} = L $, and thus Lemma \ref{lem:basic} shows that
\[
S\leq S^{*}+LM^{*}+2L\Delta\,,
\]
and Lemma \ref{lem:reas} shows that
\[
   \int_{\R_{+}} \llangle c(t)-\xi(t),2\gamma\bm{\eta}\odot\bm{\delta}\rrangle\,dt\leq2\gamma \langle \bm{\eta}, \bm{\delta}\rangle \left(S+L\Delta\right).
\]
Combining the above three equations establishes \eqref{eq:unfair-gen}.

To verify \eqref{eq:tobeusedincombi}, we note that $\zeta = 8 \gamma (\log(U) + C)$ and $L = 1$. Since $x^2 \log(1/x) \leq  (2/e) x^{3/2}$ for all $x \geq 0$, we have
\[
\langle \bm{\eta}, \bm{\delta}\rangle = \sum_{i=1}^{n} \left(\frac{u_{i}}{U}\right)^{2}\log \left(\frac{U}{u_{i}}\right)\leq\frac{2}{e}\sum_{i=1}^{n}
\left(\frac{u_{i}}{U}\right)^{3/2}\leq\frac{2}{e}\leq 1\,,
\]
where we have used $U = \sum_{i=1}^n u_i$ and $u \geq 0$.
We conclude that
\begin{align*}
S^{u}+M^{u} & \leq(\zeta+2\gamma)(S^{*}+M^{*})+(1+2\zeta+6\gamma)\Delta\\
 & \leq8\gamma(\log(U)+C+1)(S^{*}+M^{*}+4\Delta)\,.\qedhere
\end{align*}
\end{proof}

\subsection{An optimal algorithm for HSTs}
We now describe a general gluing theorem from which \pref{thm:HST} will follow directly.
In this section, we will consider the description of randomized algorithms using continuous time random states $\rho(t)$
rather than the deterministic description via the law of the random state.

Let $X$ be a tree metric where the root has $m$ children below which are subtrees $H_1, \ldots, H_m$
connected to the root by edges of lengths $w_1, \ldots, w_m > 0$.
Denote by $X_1, \ldots, X_m$ the metric spaces corresponding
to the leaves in $H_1,\ldots,H_m$, respectively.
Furthermore, let us assume that $\mathrm{diam}(X_i) \leq \frac{w_i}{4\tau}$ for some $\tau > 1$ and
every $i=1,\ldots,m$.
Applied recursively, this corresponds to the assumption that $X$ is a $4 \tau$-HST.

\begin{theorem} \label{thm:combi}
Suppose that for any $i \in [m]$ there exists an online algorithm on $X_i$ with total cost $\cost_i$ such that
\begin{equation} \label{eq:Xi}
   \cost_i \leq \frac{8\tau}{\tau-1} \left(\log(u_i) + C\right) \left(\cost^*_i + 4\,\diam(X_i)\right) \,,
\end{equation}
where $\cost^*_i$ is the total cost of the offline optimum on $X_i$ with a potentially different intial state than the online algorithm.
Then there exists an online algorithm on $X$ with total cost $\cost$ such that, with $U=\sum_{i=1}^m u_i$,
\begin{equation} \label{eq:toprovecombi}
\cost \leq \frac{8\tau}{\tau-1} \left(\log(U) + C + 1\right) \left(\cost^* + 4\, \diam(X)\right) \,,
\end{equation}
where $\cost^*$ is the total cost of the offline optimum on $X$ with a potentially different intial state than the online algorithm.
\end{theorem}

We fix a cost path on $X$ (which induces cost paths on the subspaces $X_i$) and denote $S^*, M^*$ for the total service and/movement cost of some offline
algorithm on $X$.
We use $S_i(t)$ and $M_i(t)$ for the costs of the online algorithm on $X_i$ satisfying \eqref{eq:Xi} up to time $t$.

A key ingredient in the proof is to introduce an unfair metrical task
system on a weighted star with weights $(1-\tau^{-1}) w_1, \hdots, (1-\tau^{-1}) w_m$, and unfair ratios $\gamma \seteq \frac{\tau}{\tau-1}$ and $\beta_i \seteq 8 \gamma (\log(u_i) + C)$.
We define the cost path for this unfair metrical task system by $c^u(t)(i) \seteq \frac{1}{\beta_i} \partial_t(S_i(t) + M_i(t))$. Let us denote by $(S^u)^*$ and $(M^{u})^*$ the (fair) service and movement cost of some offline algorithm on a weighted star with the cost path $c^u$. The following lemma justifies
our consideration of this setting.

\begin{lemma} \label{lem:umtstoopt}
One has
\[
(S^u)^* + (M^u)^* \leq S^* + M^* \,.
\]
\end{lemma}

\begin{proof}
Let $T_i$ be the number of times the offline algorithm
on $X$ uses the edge to the root with weight $w_i$, and consider the disjoint time intervals $I_i(k) \subset [0,T], k \in [T_i], i \in [m]$ such that for $t \in I_i(k)$ the offline algorithm's state is in $X_i$ . One clearly has
\[
   (S^u)^* + (M^u)^* \leq \sum_{i=1}^m \left((1-\tau^{-1}) w_i T_i + \sum_{k=1}^{T_i} \int_{I_i(k)} c^u(t)(i)\,dt\right) \,.
\]
Moreover using \eqref{eq:Xi},
\[
   \int_{I_i(k)} c^u(t)(i)\, dt = \int_{I_i(k)} \frac{1}{\beta_i} \partial_t(S_i(t) + M_i(t))\, dt \leq \int_{I_i(k)} \partial_t(S^*(t) + M^*(t))\, dt + 
   4\, \diam(X_i) \,.
\]
We also note that $\sum_{i,k} \int_{I_i(k)} \partial_t(S^*(t) + M^*(t))\ dt \leq S^* + M^* - \sum_{i=1}^m w_i T_i$,
and thus combined with the two preceding inequalities, we have
\[
   (S^u)^* + (M^u)^* \leq S^*+ M^* + \sum_{i=1}^m T_i \left(4\, \mathrm{diam}(X_i) - \frac{w_i}{\tau}\right).
\]
This concludes the proof, since $\mathrm{diam}(X_i) \leq \frac{w_i}{4\tau}$ by assumption.
\end{proof}

We now have all the ingredients to complete the proof of \pref{thm:combi}.

\begin{proof}[Proof of \pref{thm:combi}]
Consider the algorithm satisfying \eqref{eq:tobeusedincombi} for the unfair metrical task system described above. Our proposed algorithm to satisfy \eqref{eq:toprovecombi} chooses, at any time $t$, the subspace to play in by following the state of the unfair MTS algorithm, and then plays accordingly to the online algorithm satisfying \eqref{eq:Xi} in that subspace. Let $S$ and $M$ be the service and movement cost of this algorithm. By the definition of the unfair MTS, and \eqref{eq:tobeusedincombi}, we have
\[
   S + M \leq S^u + M^u \leq \frac{8\tau}{\tau-1} (\log(U) + C + 1) \left((S^u)^* + (M^u)^* + 4\,\diam(X)\right) \,,
\]
and thus \pref{lem:umtstoopt} concludes the proof of \eqref{eq:toprovecombi}.
\end{proof}

Finally, let us prove our main result.

\begin{proof}[Proof of \pref{thm:HST}]
   Consider a general $n$-point HST metric $(X,d)$.
   By a standard approximation argument, there is an $8$-HST metric $(X,d')$ such that $d \leq d' \leq 8\,d$ and $d'$
   can be realized as the shortest-path metric on a rooted tree $T=(V,E)$ with positive edge lengths $w : E \to (0,\infty)$,
   where $X$ constitutes the set of leaves of $T$ and, moreover, the combinatorial depth of $T$ is $D \leq O(\log n)$.

   If we now apply \pref{thm:combi} (with $\tau=2$ and $u_i = |X_i|$) recursively along $T$, we obtain an online algorithm satisfying
   \[
      \cost \leq 8 \left(\log n + C + D\right) \left(\cost^* + 4\, \diam(X)\right) \leq O(\log n) \left(\cost^* + \diam(X)\right).\qedhere
   \]
\end{proof}

\subsection*{Acknowledgements}

Part of this work was carried out while
M. B. Cohen, Y. T. Lee, and J. R. Lee were
at Microsoft Research in Redmond.
We thank Microsoft for their hospitality.
%M. B. Cohen and A. M\k{a}dry are supported by
%NSF grant CCF-1553428 and an Alfred P. Sloan Fellowship.
J. R. Lee is supported by NSF grants CCF-1616297 and CCF-1407779
and a Simons Investigator Award. 
Y. T. Lee is supported by NSF grant CCF-1740551 and CCF-1749609.

\bibliographystyle{alpha}
\bibliography{MTS}

\newcommand{\etalchar}[1]{$^{#1}$}
\begin{thebibliography}{BBMN15}

\bibitem[Bar96]{Bar96}
Yair Bartal.
\newblock Probabilistic approximations of metric spaces and its algorithmic
  applications.
\newblock In {\em 37th Annual Symposium on Foundations of Computer Science,
  {FOCS} '96, Burlington, Vermont, USA, 14-16 October, 1996}, pages 184--193,
  1996.

\bibitem[BBBT97]{BBBT97}
Yair Bartal, Avrim Blum, Carl Burch, and Andrew Tomkins.
\newblock A polylog(n)-competitive algorithm for metrical task systems.
\newblock In {\em Proceedings of the Twenty-ninth Annual ACM Symposium on
  Theory of Computing}, STOC '97, pages 711--719, New York, NY, USA, 1997. ACM.

\bibitem[BBM06]{BBM06}
Yair Bartal, B\'ela Bollob\'as, and Manor Mendel.
\newblock Ramsey-type theorems for metric spaces with applications to online
  problems.
\newblock {\em J. Comput. System Sci.}, 72(5):890--921, 2006.

\bibitem[BBMN15]{BBMN15}
Nikhil Bansal, Niv Buchbinder, Aleksander Madry, and Joseph Naor.
\newblock A polylogarithmic-competitive algorithm for the {$k$}-server problem.
\newblock {\em J. ACM}, 62(5):Art. 40, 49, 2015.

\bibitem[BBN10]{BBN10}
Nikhil Bansal, Niv Buchbinder, and Joseph Naor.
\newblock Metrical task systems and the k-server problem on hsts.
\newblock In {\em Proceedings of the 37th International Colloquium Conference
  on Automata, Languages and Programming}, ICALP'10, pages 287--298, Berlin,
  Heidelberg, 2010. Springer-Verlag.

\bibitem[BBN12]{BBN12}
Nikhil Bansal, Niv Buchbinder, and Joseph Naor.
\newblock A primal-dual randomized algorithm for weighted paging.
\newblock {\em J. ACM}, 59(4):Art. 19, 24, 2012.

\bibitem[BCL{\etalchar{+}}18]{BCLLM18}
S{\'{e}}bastien Bubeck, Michael~B. Cohen, Yin~Tat Lee, James~R. Lee, and
  Aleksander Madry.
\newblock k-server via multiscale entropic regularization.
\newblock In {\em Proceedings of the 50th Annual {ACM} {SIGACT} Symposium on
  Theory of Computing, {STOC} 2018, Los Angeles, CA, USA, June 25-29, 2018},
  pages 3--16, 2018.

\bibitem[BKRS00]{BKRS00}
Avrim Blum, Howard Karloff, Yuval Rabani, and Michael Saks.
\newblock A decomposition theorem for task systems and bounds for randomized
  server problems.
\newblock {\em SIAM J. Comput.}, 30(5):1624--1661, 2000.

\bibitem[BLMN05]{BLMN05}
Yair Bartal, Nathan Linial, Manor Mendel, and Assaf Naor.
\newblock On metric {R}amsey-type phenomena.
\newblock {\em Ann. of Math. (2)}, 162(2):643--709, 2005.

\bibitem[BLS92]{BLS92}
Allan Borodin, Nathan Linial, and Michael~E. Saks.
\newblock An optimal on-line algorithm for metrical task system.
\newblock {\em J. ACM}, 39(4):745--763, October 1992.

\bibitem[FM03]{FM03}
Amos Fiat and Manor Mendel.
\newblock Better algorithms for unfair metrical task systems and applications.
\newblock {\em SIAM Journal on Computing}, 32(6):1403--1422, 2003.

\bibitem[FRT04]{FRT04}
Jittat Fakcharoenphol, Satish Rao, and Kunal Talwar.
\newblock A tight bound on approximating arbitrary metrics by tree metrics.
\newblock {\em J. Comput. Syst. Sci.}, 69(3):485--497, 2004.

\bibitem[Leo09]{Leo09}
Giovanni Leoni.
\newblock {\em A first course in Sobolev spaces}.
\newblock American Mathematical Soc., 2009.

\bibitem[Sei99]{Sei99}
Steve Seiden.
\newblock Unfair problems and randomized algorithms for metrical task systems.
\newblock {\em Inf. Comput.}, 148(2):219--240, February 1999.

\end{thebibliography}

\end{document}